\documentclass[10pt,a4paper]{article}
\usepackage{amsfonts}
\usepackage[dvips]{graphicx}
\usepackage{amsmath}
\usepackage{mathrsfs}
\usepackage{makeidx}
\usepackage{xypic}
\usepackage[nottoc]{tocbibind}
\usepackage{pstricks}
\usepackage{bbm}
\usepackage[arrow, matrix, curve]{xy}
\usepackage{amsthm}
\usepackage{amssymb}
\usepackage{epic}
\usepackage{eepic}
\usepackage{times}
\usepackage{epsfig}

\xymatrixcolsep{0.5cm}
\xymatrixrowsep{0.5cm}
\newtheorem{theorem}{Theorem}[section]
\newtheorem{proposition}{Proposition}[section]
\newtheorem{lemma}{Lemma}[section]

\newtheorem{corollary}{Corollary}[section]

\newcommand{\beqa}{\begin{eqnarray}}
\newcommand{\eeqa}{\end{eqnarray}}
\newcommand{\rf}[1]{(\ref{#1})}

\newcommand{\la}{\lambda}

\makeindex
\numberwithin{equation}{section}

\begin{document}

\begin{flushright}
YITP-SB-12-20
\end{flushright}

\bigskip

\bigskip

\begin{center}
\textbf{\Large Form factors and complete spectrum of XXX antiperiodic higher spin chains by quantum separation of variables}

\vspace{50pt}

{\large G.~Niccoli\footnote{{\large {\small YITP, Stony Brook University,
New York 11794-3840, USA, niccoli@max2.physics.sunysb.edu}}}}

\vspace{50pt}

\vspace{80pt}
\end{center}

\begin{itemize}
\item[] \textbf{{Abstract}}\,\,\,The antiperiodic
transfer matrix associated to higher spin representations of the rational
6-vertex Yang-Baxter algebra is analyzed by generalizing the approach
introduced recently in \cite{ARXFGMN12-SG}, for the cyclic representations, in 
\cite{ARXFN12-0}, for the spin-1/2 highest weight representations, and in \cite{ARXFN12-2}, for the spin 1/2 representations of the reflection algebra. Here, we
derive the complete characterization of the transfer matrix spectrum and we
prove its simplicity in the framework of Sklyanin's quantum separation of
variables (SOV). Then, the characterization of local operators by Sklyanin's
quantum separate variables and
the expression of the scalar products of \textit{separates states} by
determinant formulae allow to compute the form factors of the local spin
operators by one determinant formulae similar to the scalar product ones.
Finally, let us comment that these results represent the SOV analogous in the antiperiodic higher spin XXX quantum chains of the results obtained for
the periodic chains in \cite{ARXFCM07} in the framework of the algebraic Bethe
ansatz.
\end{itemize}

\newpage
\tableofcontents
\newpage

\section{Introduction}

The lattice quantum integrable models associated by the quantum inverse
scattering method (QISM) \cite{ARXFSF78}-\cite{ARXFIK82} to the representations of
the rational 6-vertex Yang-Baxter algebra on higher spin-$s$ quantum chains (%
$s$ any positive half-integer) with antiperiodic boundary conditions are analyzed. Under the homogeneous limit, such analysis allows to describe the
spin-$s$ XXX quantum chain with antiperiodic boundary conditions. Then, it
is worth recalling that under periodic boundary conditions\footnote{%
See \cite{ARXFZam80}-\cite{ARXFBab85}, \cite{ARXFKS82} and \cite{ARXFKRS} for the first
studies of the integrable spin-$s$ XXZ quantum spin chains.
The spin-1 Hamiltonian was first given in \cite{ARXFZam80}; Bethe ansatz
equations for the spin-$s$ models with periodic boundary conditions were
obtained and studied in \cite{ARXFSg84,ARXFBab85,ARXFKR}, see also \cite{ARXFAlz89}-\cite{ARXFKlp91}.} these quantum models have been investigated by using the algebraic
Bethe ansatz\footnote{%
See \cite{ARXFSF78}-\cite{ARXFFST80} and reference therein.} (ABA) and
results are known both for the spectrum characterization and for the correlation
functions \cite{ARXFK01,ARXFCM07}. 
In the papers \cite{ARXFMRM05,ARXFGM05}, the analysis of the spectrum under general toroidal boundary conditions has been developed thanks to the extension\footnote{In fact, the analysis developed in these papers also generalizes the (nested) ABA for the integrable quantum models associated to rational higher rank Yang-Baxter algebras for these general boundary conditions.} to these quantum models of the ABA method. It is also relevant to comment that similarly, by using the Baxter's gauge transformation technique \cite{ARXFBaxBook}, the spectrum of the open XXX quantum spin chains with general non-diagonal integrable boundary conditions has been analyzed by ABA in \cite{ARXFRMG03,ARXFRM05} and also by the functional version of Sklyanin's quantum separation of variables (SOV) \cite{ARXFSk1}-\cite{ARXFSk3} in \cite{ARXFFSW08,ARXFFGSW11}. 

Here, we show how to apply an approach based on Sklyanin's SOV to the special case of antiperiodic boundary conditions\footnote{Let us recall that in \cite{ARXFFram+2} has been developed the analysis by the functional SOV of the related but more general spin-boson model introduced and first analyzed by ABA in \cite{ARXFFram+1}.}, which allows to
achieve the complete characterization of the spectrum and the computation of the
matrix elements of local operators on the transfer matrix eigenstates. It is natural to consider the analysis here presented as the generalization to the SOV-framework of the Lyon group
method \cite{ARXFKitMT99}-\cite{ARXFKKMST07} implemented in the ABA framework\footnote{See \cite{ARXFKKMNST07}-\cite{ARXFKKMNST08} for the extension of this method and of the corresponding results to the open spin 1/2 quantum chains with diagonal
boundary conditions.}. This approach has been first developed\footnote{%
Let us comment that these papers use as required setup the series of papers 
\cite{ARXFNT-10}-\cite{ARXFGN12} where the complete spectrum has
been characterized for the lattice quantum sine-Gordon model and for the $%
\tau _{2}$-model and the chiral Potts model \cite{ARXFBS90}-\cite{ARXFTarasovSChP},
respectively.} in \cite{ARXFGMN12-SG,ARXFGMN12-T2} for the lattice quantum
sine-Gordon model \cite{ARXFFST80,ARXFIK82}\ and for the $\tau _{2}$-model\footnote{%
See the series of works \cite{ARXFGIPS06}-\cite{ARXFGIPS09} for previous analysis by
SOV method of the $\tau _{2}$-model.} \cite{ARXFBa04}. In particular, in \cite{ARXFGMN12-SG,ARXFGMN12-T2} the reconstruction of local operators by quantum
separate variables and one determinant formulae for the scalar products
of separate states\footnote{%
See Section \ref{ARXFSP} for the definition of these states in our current model.%
} have been derived and used to compute matrix elements of local operators
in determinant form. Further key quantum integrable models have been analyzed
by this approach getting the same type of universal results. In \cite{ARXFN12-0} the antiperiodic\footnote{Let us comment that previous results for this model under antiperiodic
boundary condition were based on the Baxter Q-operator \cite{ARXFBBOY95} and the
functional separation of variables of Sklyanin developed first in \cite{ARXFSk2} for the XXX spin chain and later in \cite{ARXFNWF09} for the XXZ spin chain.}
XXZ spin 1/2 quantum chain \cite{ARXFH28}-\cite{ARXFLM66} has been considered while for the spin 1/2
representations of the reflection algebra \cite{ARXFGau71}-\cite{ARXFGZ94} with
non-diagonal boundaries the SOV setup has been implemented in \cite{ARXFN12-2} and there also the
matrix elements of some string of local operators have been computed. 

Let us finally comment that our interest toward the SOV method of Sklyanin
is due to the fact that it allows to overcome several problems which affect
others methods like for example the coordinate Bethe ansatz \cite{ARXFBe31}, 
\cite{ARXFBaxBook} and \cite{ARXFABBBQ87}, the Baxter Q-operator method \cite{ARXFBaxBook}, the algebraic Bethe ansatz \cite{ARXFSF78}-\cite{ARXFFST80}, the analytic
Bethe ansatz \cite{ARXFRe83-1}-\cite{ARXFRe83-2}. Indeed, SOV applies for a large
class of integrable quantum models to which others Bethe ansatz methods do
not apply; both the eigenvalues and the eigenstates are constructed and
under simple conditions these spectrum characterizations are complete. Moreover, in all the models analyzed in the series of papers \cite{ARXFNT-10}-\cite{ARXFGN12}, \cite{ARXFN12-0}, \cite{ARXFN12-2} and \cite{ARXFN12-3}
in the SOV framework the non-degeneracy of the
transfer matrix spectrum has been proven.

\section{Antiperiodic 6-vertex models}

\subsection{Higher spin representations}

Let $S_{n}^{\pm }$ and $S_{n}^{z}$ be the generators of the $sl(2)$ algebra:%
\begin{equation}
\lbrack S^{z},S^{\pm }]=\pm S^{\pm },\text{ \ }[S^{+},S^{-}]=2S^{z},
\end{equation}%
and let R$^{(s_{n})}\simeq $ $\mathbb{C}^{2s_{n}+1}$\ be linear spaces
(local quantum spaces) of dimension $(2s_{n}+1)$ with $2s_{n}\in \mathbb{Z}%
^{>0}$. A spin-$s_{n}$ representation of the $sl(2)$ algebra is associated
to any linear space R$^{(s_{n})}$ by defining:%
\begin{equation}
S_{n}^{z}=\text{diag}(s_{n},s_{n}-1,\ldots ,-s_{n}),\text{ \ \ }%
S_{n}^{+}=\left( S_{n}^{-}\right) ^{t}=\left( 
\begin{array}{llll}
0 & x_{n}(1) &  &  \\ 
& \ddots & \ddots &  \\ 
&  & \ddots & x_{n}(2s_{n}) \\ 
&  &  & 0%
\end{array}%
\right) ,
\end{equation}%
where $x_{n}(j)\equiv \sqrt{j(2s_{n}+1-j)}$. Then, to each local quantum
space R$^{(s_{n})}$ is associated a so-called Lax operator:%
\begin{equation}
\mathsf{L}_{0n}^{(1/2,s_{n})}(\lambda )\equiv \left( 
\begin{array}{cc}
\lambda +\eta (1/2+S_{n}^{z}) & \eta S_{n}^{-} \\ 
\eta S_{n}^{+} & \lambda +\eta (1/2-S_{n}^{z})%
\end{array}%
\right) _{0}\in \text{End}(V_{0}^{(1/2)}\otimes V_{n}^{(s_{n})}),
\end{equation}%
which satisfies the Yang-Baxter equation:%
\begin{equation}
R_{12}(\lambda -\mu )\mathsf{L}_{1n}^{(1/2,s_{n})}(\lambda )\mathsf{L}%
_{2n}^{(1/2,s_{n})}(\mu )=\mathsf{L}_{2n}^{(1/2,s_{n})}(\mu )\mathsf{L}%
_{1n}^{(1/2,s_{n})}(\lambda )R_{12}(\lambda -\mu ),  \label{ARXFYBA}
\end{equation}%
w.r.t. the rational 6-vertex R-matrix:%
\begin{equation}
R_{12}(\lambda )=\mathsf{L}_{12}^{(1/2,1/2)}(\lambda )\equiv \left( 
\begin{array}{cccc}
\lambda +\eta & 0 & 0 & 0 \\ 
0 & \lambda & \eta & 0 \\ 
0 & \eta & \lambda & 0 \\ 
0 & 0 & 0 & \lambda +\eta%
\end{array}%
\right) .
\end{equation}%
Now, we are in the position to define the so-called monodromy matrix:%
\begin{equation}
\mathsf{M}_{0}^{(1/2)}(\lambda )\equiv \left( 
\begin{array}{cc}
\mathsf{A}(\lambda ) & \mathsf{B}(\lambda ) \\ 
\mathsf{C}(\lambda ) & \mathsf{D}(\lambda )%
\end{array}%
\right) \equiv \mathsf{L}_{0\mathsf{N}}^{(1/2,s_{\mathsf{N}})}(\lambda -\eta
_{\mathsf{N}})\cdots \mathsf{L}_{01}^{(1/2,s_{1})}(\lambda -\eta _{1})\in 
\text{End}(\text{R}_{0}^{(1/2)}\otimes _{n=1}^{\mathsf{N}}\text{R}%
_{n}^{(s_{n})}),
\end{equation}%
where the $\eta _{n}$ are the so-called inhomogeneities parameters. Then the
monodromy matrix $\mathsf{M}_{0}^{(1/2)}(\lambda )$ is itself solution of
the Yang-Baxter equation:%
\begin{equation}
R_{12}(\lambda -\mu )\mathsf{M}_{1}^{(1/2)}(\lambda )\mathsf{M}%
_{2}^{(1/2)}(\mu )=\mathsf{M}_{2}^{(1/2)}(\mu )\mathsf{M}_{1}^{(1/2)}(%
\lambda )R_{12}(\lambda -\mu ),
\end{equation}%
w.r.t. the rational 6-vertex R-matrix.

\subsubsection{Yang-Baxter algebra representations on vector and covector
spaces\label{ARXFYB-LR}}

Let $|k,n\rangle $ be a vector in R$_{n}^{(s_{n})}$ characterized by:%
\begin{equation}
S_{n}^{z}|k,n\rangle =(k-s_{n})|k,n\rangle ,\text{ \ }k\in
\{0,1,...,2s_{n}\},
\end{equation}%
i.e. the set of $|k,n\rangle $ defines a $S_{n}^{z}$-eigenbasis of the local
space R$_{n}^{(s_{n})}$, and let us denote with L$_{n}^{(s_{n})}$ the linear space
dual of R$_{n}^{(s_{n})}$ and with $\langle k,n|$ a covector defined by:%
\begin{equation}
\langle k,n|k^{\prime },n\rangle \equiv \delta _{k,k^{\prime }}\text{ \ \ }%
\forall k,k^{\prime }\in \{0,...,2s_{n}\},
\end{equation}%
i.e. the covectors $\langle k,n|$ define the $S_{n}^{z}$-eigenbasis in the
dual linear space L$_{n}$. We can naturally introduce a scalar
product w.r.t. the covector-vector basis by setting:
\begin{equation}
(|k,n\rangle ,|k^{\prime },n\rangle )\equiv \langle k,n|k^{\prime },n\rangle.
\end{equation}%
In the \textit{left} (covectors) and \textit{right} (vectors) linear spaces:%
\begin{equation}
\mathcal{L}_{\mathsf{N}}\equiv \otimes _{n=1}^{\mathsf{N}}\text{L}%
_{n}^{(s_{n})},\text{ \ \ \ \ }\mathcal{R}_{\mathsf{N}}\equiv \otimes
_{n=1}^{\mathsf{N}}\text{R}_{n}^{(s_{n})},
\end{equation}%
the representations of the local $sl(2)$ generators induce left and right
spin-$\{s_{1},...,s_{\mathsf{N}}\}$\ representations of dimension $d_{%
\mathsf{N}}\equiv \prod_{n=1}^{\mathsf{N}}(2s_{n}+1)$ with $\mathsf{N}$
inhomogeneities of the rational 6-vertex Yang-Baxter algebra; i.e. the quadratic algebra defined by the set of commutation relations of the monodromy matrix
elements $\mathsf{A}(\lambda )$, $\mathsf{B}(\lambda )$, $\mathsf{C}(\lambda
)$ and $\mathsf{D}(\lambda )$ encoded in the Yang-Baxter equation satisfied
by $\mathsf{M}_{0}^{(1/2)}(\lambda )$.

\subsubsection{Rational 6-vertex higher spin transfer matrices}

Let us remark that the rational 6-vertex $R$-matrix satisfies the following $%
GL(2,\mathbb{C})$ symmetry: 
\begin{equation}
R_{12}(\lambda )W_{1}\otimes W_{2}=W_{2}\otimes W_{1}R_{12}(\lambda ),
\label{ARXFGL2-Sym}
\end{equation}%
where $W$ is any invertible $2\times 2$ matrix. Then, we have that for any $%
W $ we can define a monodromy matrix:%
\begin{equation}
\mathsf{M}_{0}^{(1/2,W)}(\lambda )\equiv W_{0}\mathsf{M}_{0}^{(1/2)}(\lambda
),
\end{equation}%
which is a solution of the Yang-Baxter equations (\ref{ARXFYBA}) w.r.t. the same
rational 6-vertex $R$-matrix. This imply that the transfer matrix:%
\begin{equation}
\mathsf{T}^{(W)}(\lambda )=\text{tr}_{0}[W_{0}\mathsf{M}_{0}^{(1/2)}(\lambda
)],
\end{equation}%
defines a one-parameter family of commuting operators. Moreover, let us
recall that the quantum determinant:%
\begin{equation}
\det_{q}\mathsf{M}_{0}^{(1/2,W)}(\lambda )\,\equiv \det W\det_{q}\mathsf{M}%
_{0}^{(1/2)}(\lambda )  \label{ARXFq-det-f}
\end{equation}%
where%
\begin{equation}
\det_{q}\mathsf{M}_{0}^{(1/2,\sigma ^{x})}(\lambda )\,\equiv \,\mathsf{B}%
(\lambda )\mathsf{C}(\lambda /q)-\mathsf{A}(\lambda )\mathsf{D}(\lambda /q)
\end{equation}%
is a central element\footnote{%
See \cite{ARXFIK81} and \cite{ARXFIK09} for an historical note.} of the Yang-Baxter
algebra (\ref{ARXFYBA}) which explicit reads:%
\begin{equation}
\det_{q}\mathsf{M}_{0}^{(1/2,\sigma ^{x})}(\lambda )\,\equiv \det_{q}\mathsf{%
\bar{M}}^{(1/2)}(x)=a(\lambda )d(\lambda -\eta ),
\end{equation}%
where\footnote{%
Note that $a(\lambda )$ has zeros at $\lambda =\eta _{n}^{-}-s_{n}\eta $,
and $d(\lambda )$ has zeros at $\lambda =\eta _{n}^{-}+s_{n}\eta $ for all $%
j=1,\ldots,$$\mathsf{N}$.}:%
\begin{equation}
a(\lambda )=-\prod_{n=1}^{\mathsf{N}}\left( \lambda -\eta _{n}^{-}+s_{n}\eta
\right) ,\text{ \ \ \ \ \ }d(\lambda )=\prod_{n=1}^{\mathsf{N}}\left(
\lambda -\eta _{n}^{-}-s_{n}\eta \right) ,
\end{equation}%
and we have used the notation $\la^{\pm }\equiv \la\pm \eta /2$. 

In the following, we solve the spectral problem for the quantum integrable
models characterized in the framework of the quantum inverse scattering
method by the following antiperiodic transfer matrix:%
\begin{equation}
\mathsf{\bar{T}}(\lambda )\equiv \mathsf{B}(\lambda )+\mathsf{C}(\lambda )=%
\mathsf{T}^{(W=\sigma ^{x})}(\lambda ).
\end{equation}%

\begin{lemma}
\textsf{I)} If $\eta \in i\mathbb{R}$ and \{$%
\eta _{1},...,\eta _{\mathsf{N}}$\}$\in \mathbb{R}^{\mathsf{N}}$, $\mathsf{%
\bar{T}}(\lambda )$ is a one parameter family of normal operators and:%
\begin{equation}
i\mathsf{\bar{T}}(\lambda )
\end{equation}%
is self-adjoint for any $\lambda \in \mathbb{C}$ such that $\lambda -\lambda
^{\ast }+\eta =0$.

\textsf{II)} If $\eta \in \mathbb{R}$ and \{$\eta
_{1},...,\eta _{\mathsf{N}}$\}$\in \left( i\mathbb{R}\right) ^{\mathsf{N}}$, 
$\mathsf{\bar{T}}(\lambda )$ is a one parameter family of normal operators
and:%
\begin{equation}
i^{\mathsf{e}_{\mathsf{N}}}\mathsf{\bar{T}}(\lambda ),
\end{equation}%
\ \ \ where $\mathsf{e}_{\mathsf{N}}=\{1$ \ for \textsf{N\ \ }%
even,\thinspace $0$ \ for \textsf{N\ \ }odd$\}$, is self-adjoint for any $%
\lambda \in \mathbb{C}$ such that $\lambda +\lambda ^{\ast }+\eta =0$.
\end{lemma}

\begin{proof}
\textsf{I)} $\mathsf{L}_{0n}^{(1/2,s_{n})}(\lambda )$ satisfies the
following Hermitian conjugation property:%
\begin{equation}
\mathsf{L}_{0n}^{(1/2,s_{n})}(\lambda )^{\dagger }\equiv \sigma _{0}^{y}%
\mathsf{L}_{0n}^{(1/2,s_{n})}(\lambda ^{\ast }-\eta )\sigma _{0}^{y},
\end{equation}%
as it can be verified by direct
computations. Here $\dagger $ implements the transposition on the local quantum space $n$
and the complex conjugation, then it holds:%
\begin{equation}
\mathsf{M}(\lambda )^{\dagger }\equiv \left( 
\begin{array}{cc}
\mathsf{A}^{\dagger }(\lambda ) & \mathsf{B}^{\dagger }(\lambda ) \\ 
\mathsf{C}^{\dagger }(\lambda ) & \mathsf{D}^{\dagger }(\lambda )%
\end{array}%
\right) =\left( 
\begin{array}{cc}
\mathsf{D}(\lambda ^{\ast }-\eta ) & -\mathsf{C}(\lambda ^{\ast }-\eta ) \\ 
-\mathsf{B}(\lambda ^{\ast }-\eta ) & \mathsf{A}(\lambda ^{\ast }-\eta )%
\end{array}%
\right) ,  \label{ARXFHerm-1}
\end{equation}%
for \{$\eta _{1},...,\eta _{\mathsf{N}}$\}$\in \mathbb{R}^{\mathsf{N}}$ and
so $\bar{\mathsf{T}}(\lambda )$ is normal for any $\lambda \in \mathbb{C}$
and the statements in \textsf{I)} follow.

\textsf{II)} It is simple to observe that:%
\begin{equation}
\mathsf{L}_{0n}^{(1/2,s_{n})}(\lambda )^{\dagger }\equiv -\sigma _{0}^{y}%
\mathsf{L}_{0n}^{(1/2,s_{n})}(-(\lambda ^{\ast }+\eta ))\sigma _{0}^{y},
\end{equation}%
then, it holds:%
\begin{equation}
\mathsf{M}(\lambda )^{\dagger }\equiv \left( 
\begin{array}{cc}
\mathsf{A}^{\dagger }(\lambda ) & \mathsf{B}^{\dagger }(\lambda ) \\ 
\mathsf{C}^{\dagger }(\lambda ) & \mathsf{D}^{\dagger }(\lambda )%
\end{array}%
\right) =\left( -1\right) ^{\mathsf{N}}\left( 
\begin{array}{cc}
\mathsf{D}(-(\lambda ^{\ast }+\eta )) & -\mathsf{C}(-(\lambda ^{\ast }+\eta
)) \\ 
-\mathsf{B}(-(\lambda ^{\ast }+\eta )) & \mathsf{A}(-(\lambda ^{\ast }+\eta
))%
\end{array}%
\right) ,  \label{ARXFHerm-2}
\end{equation}%
for $\{\eta _{1},...,\eta _{\mathsf{N}}\}\in \left( i\mathbb{R}\right) ^{%
\mathsf{N}}$ and so $\mathsf{\bar{T}}(\lambda )$ is normal for any $\lambda
\in \mathbb{C}$ and the statements in \textsf{II)} follow.
\end{proof}

\subsection{Antiperiodic 6-vertex quantum integrable higher spin chains}

\subsubsection{Fusion procedure for 6-vertex representations}

The fusion procedure was first developed in \cite{ARXFKRS} for the case of the
rational 6-vertex representations of the type analyzed here and later in 
\cite{ARXFKR} for the trigonometric ones. Our interest in the fusion procedure
is related to its use to reconstruct local operators in terms of the
Yang-Baxter algebra generators and the fused transfer matrix. This result
was first derived in the case of the periodic transfer matrix ($W=\mathbb{I}%
_{2\times 2}$) in \cite{ARXFMaiT00}; here we will extend that result in the case
of the antiperiodic $(W=\sigma ^{x})$\ transfer matrices.

The fusion procedure can be used to construct monodromy matrices with
auxiliary spaces of dimension higher than 2 starting from the one with
2-dimensional auxiliary space. We illustrate this
procedure in the antiperiodic case following a presentation similar to that of \cite{ARXFYB95}. Let us remark that the following
commutation relations hold:%
\begin{equation}
\lbrack \mathsf{L}_{0a}^{(1/2,s_{a})}(\lambda ),\sigma _{0}^{x}\otimes
\Sigma _{a}^{x}]=0,\text{ \ \ }\forall a\in \{1,...,\mathsf{N}\}
\label{ARXFSym-s-1/2}
\end{equation}%
where $\Sigma _{a}^{x}$ is the $(2s_{a}+1)\times (2s_{a}+1)$ matrix with
elements all zeros except those along the antidiagonal which are 1. The
above property is a consequence of the commutation relations:%
\begin{equation}
S_{a}^{\pm }\Sigma _{a}^{x}=\Sigma _{a}^{x}S_{a}^{\mp },\text{ \ \ }%
S_{a}^{z}\Sigma _{a}^{x}=-\Sigma _{a}^{x}S_{a}^{z}
\end{equation}%
where $S_{a}^{\pm }$ and $S_{a}^{z}$ are the generators of the spin-$%
(2s_{a}+1)$ representation of $sl(2)$. Form (\ref{ARXFSym-s-1/2}), we have also:%
\begin{equation}
\lbrack \mathsf{M}_{0}^{(1/2)}(\lambda ),\sigma _{0}^{x}\otimes \Sigma
^{x}]=0,\text{ \ \ }\forall a\in \{1,...,\mathsf{N}\}
\end{equation}%
where:%
\begin{equation}
\Sigma ^{x}\equiv \otimes _{a=1}^{\mathsf{N}}\Sigma _{a}^{x}.
\end{equation}%
Let us now define the antiperiodic monodromy matrices:%
\begin{equation}
\mathsf{\bar{M}}_{0}^{(1/2)}(\lambda )\equiv \sigma _{0}^{x}\mathsf{M}%
_{0}^{(1/2)}(\lambda )\in \text{End}(\text{R}_{0}^{(1/2)}\otimes \mathcal{R}%
_{\mathsf{N}}),\text{ \ }\mathsf{\bar{M}}_{a}^{(s_{a})}(\lambda )\equiv
\Sigma _{a}^{x}\mathsf{M}_{a}^{(s_{a})}(\lambda )\in \text{End}(\text{R}%
_{a}^{(s_{a})}\otimes \mathcal{R}_{\mathsf{N}}),
\end{equation}%
and denoted with R$_{\langle 12\rangle }^{(s)}$ and R$_{(12)}^{(s-1)}$ the
linear spaces of dimension $2s+1$ and $2s-1$ defined by the following
decomposition 
\begin{equation}
\text{R}_{1}^{(s-\frac{1}{2})}\otimes \text{R}_{2}^{(\frac{1}{2})}\simeq 
\text{R}_{\langle 12\rangle }^{(s)}\oplus \text{R}_{(12)}^{(s-1)},
\end{equation}%
of the tensor product of R$_{2}^{(\frac{1}{2})}$ and R$_{1}^{(s-\frac{1}{2}%
)} $, linear spaces of dimension $2$ and $2s$, respectively. This tensor
product decomposition allows to define the operator:%
\begin{equation}
P_{12}=P_{\langle 12\rangle }^{+}\oplus P_{(12)}^{-}\,,
\end{equation}%
as the direct sum of the projector $P_{\langle 12\rangle }^{+}$ in R$%
_{\langle 12\rangle }^{(s)}$ and the projector $P_{(12)}^{-}$ in R$%
_{(12)}^{(s-1)}$. Then the fusion of the monodromy matrices \ $\mathsf{\bar{M%
}}_{a}^{(s_{a})}(\lambda )$ has the same form as for the periodic case and
it reads:%
\begin{equation}
P_{12}\mathsf{\bar{M}}_{1}^{({1/2})}(\lambda ^{-}+s\eta )\mathsf{\bar{M}}%
_{2}^{(s-1/2)}(\lambda ^{-})P_{12}=\left( \!\!%
\begin{array}{cc}
\mathsf{\bar{M}}_{\langle 12\rangle }^{(s)}(\lambda ) & 0 \\* 
\ast & \mathsf{\bar{M}}_{(12)}^{(s-1)}(\lambda -\eta )\det_{q}\mathsf{\bar{M}%
}^{(1/2)}(\lambda +(s-1)\eta )%
\end{array}%
\!\!\right) ,  \label{ARXFfusion}
\end{equation}%
where now the quantum determinant of the antiperiodic monodromy matrix $%
\mathsf{\bar{M}}^{(1/2)}(\lambda )$ appears. Note that to any (higher)
monodromy matrix $\mathsf{\bar{M}}_{0}^{(s)}(\lambda )$ we can associate the
one-parameter family of higher transfer matrix:%
\begin{equation}
\mathsf{\bar{T}}^{(s)}(\lambda )=\text{tr}_{0}\mathsf{\bar{M}}%
_{0}^{(s)}(\lambda )\in \text{End}(\mathcal{R}_{\mathsf{N}})\text{ \ }%
\forall \lambda \in \mathbb{C},
\end{equation}%
which define commuting families of operators:%
\begin{equation}
\lbrack \mathsf{\bar{T}}^{(s_{1})}(\lambda ),\mathsf{\bar{T}}^{(s_{2})}(\mu
)]=0\text{ \ \ }\forall 2s_{1},2s_{2}\in \mathbb{Z}^{>0}.
\end{equation}%
The formula (\ref{ARXFfusion}) implies in particular the following recursion relations: 
\begin{equation}
\mathsf{\bar{T}}^{(s)}(\lambda )=\mathsf{\bar{T}}^{(1/2)}(\lambda ^{-}+s\eta
)\mathsf{\bar{T}}^{(s-1/2)}(\lambda ^{-})-\det_{q}\mathsf{\bar{M}}%
^{(1/2)}(\lambda +(s-1)\eta )\mathsf{\bar{T}}^{(s-1)}(\lambda -\eta ),
\label{ARXFTrans-fus}
\end{equation}%
for antiperiodic higher transfer matrices.

\subsubsection{Antiperiodic higher spin XXX quantum chains}

It is worth pointing out that the analysis of the antiperiodic
transfer matrix $\mathsf{\bar{T}}^{(1/2)}(\lambda )$ allows in particular to
describe the XXX higher spin-$s$ quantum chains in the special case of the
homogeneous limit ($\eta _{n}\rightarrow 0$) and under homogeneous spin-$s$
representations ($s_{n}=s$ for any $n\in \{1,...,\mathsf{N}\}$). Indeed, the
Hamiltonian of the XXX spin-$s$ quantum chain is obtain by logarithmic
derivative of the transfer matrix associated to the fundamental monodromy
matrix $\mathsf{\bar{M}}^{(s)}(\lambda )$:%
\begin{equation}
\left. H=\mathsf{\bar{T}}^{(s)}(\lambda )^{-1}\frac{d}{d\lambda }\mathsf{%
\bar{T}}^{(s)}(\lambda )\right\vert _{\lambda =0}, \label{ARXFham}
\end{equation}%
with the following boundary conditions:%
\begin{equation}
S_{\mathsf{N}+1}^{z}=-S_{1}^{z},\hspace{20pt}S_{\mathsf{N}+1}^{\pm
}=S_{1}^{\mp }.
\end{equation}%
Note that the transfer matrix $\mathsf{\bar{T}}^{(s)}(\lambda )$ used to
construct the XXX spin-$s$ Hamiltonian is obtained in terms of the transfer
matrix $\mathsf{\bar{T}}(\lambda )$ by using the recursion relations (\ref{ARXFTrans-fus}). This point together with the simplicity of the spectrum of the
transfer matrix $\mathsf{\bar{T}}(\lambda )$ (which we will show in the following) implies that it is enough to characterize the spectrum of this last
transfer matrix to have in particular the solution of the XXX spin-$s$
quantum chain.

\section{SOV-representations}

\label{ARXFSOV-Gen}A separation of variable (SOV) representation \cite{ARXFSk1,ARXFSk2,ARXFSk3} for the $\mathsf{\bar{T}}$-spectral problem is associated to a
representation for which the commutative family of operators $\mathsf{D}%
(\lambda )$ (or $\mathsf{A}(\lambda )$) is diagonal and with simple spectrum.

\begin{theorem}
If the inhomogeneities $\{\eta _{1},...,\eta _{\mathsf{N}}\}\in \mathbb{C}$ $%
^{\mathsf{N}}$ satisfy the conditions:
\begin{equation}
\eta _{a}\neq \eta _{b}\mathsf{ \,\,\, mod\,}\eta \,\,\,\,\,\forall a\neq b\in
\{1,...,\mathsf{N}\}  \label{ARXFE-SOV}
\end{equation}%
then $\mathsf{D}(\lambda )$\ and $\mathsf{A}(\lambda )$ are diagonalizable
and with simple spectrum and the $\mathsf{\bar{T}}$-spectral
problem admits separate variable representations.
\end{theorem}In the next subsection we construct explicitly the $\mathsf{D}$-eigenbasis in this
way proving the Theorem \ref{ARXFE-SOV}.
\subsection{Construction of SOV-representation in $\mathsf{D}$-eigenbasis}
 Let%
\begin{equation}
\langle 0|\equiv \otimes _{n=1}^{\mathsf{N}}\langle 1,n|\text{ \ \ \ and \ \ 
}|0\rangle \equiv \otimes _{n=1}^{\mathsf{N}}|1,n\rangle ,
\end{equation}%
be the left (covector) and right (vector) \textit{references states}, where:%
\begin{equation}
\langle 1,n|=\left( 1,0,...,0\right) _{1,2s_{n}+1},\text{ \ \ }|1,n\rangle
=\left( 
\begin{array}{c}
1 \\ 
0 \\ 
\vdots \\ 
0%
\end{array}%
\right) _{2s_{n}+1,1},
\end{equation}%
then:

\begin{theorem}
\textsf{I)} \underline{Left $\mathsf{D}(\lambda )$ SOV-representations} \
If $\left( \ref{ARXFE-SOV}\right) $ are verified, the states $\langle $\textbf{h}$|\equiv \langle h_{1},...,h_{\mathsf{N}}|$, defined by:
\begin{equation}
\langle \text{\textbf{h}}|\equiv \frac{1}{\text{\textsc{n}}}\langle
0|\prod_{n=1}^{N}\prod_{k_{n}=0}^{h_{n}-1}\frac{C(\eta _{n}^{(k_{n})})}{%
d(\eta _{n}^{(k_{n}+1)})},  \label{ARXFD-left-eigenstates}
\end{equation}%
where%
\begin{equation}
\text{\textsc{n}}=\prod_{1\leq b<a\leq \mathsf{N}}(\eta _{a}^{(0)}-\eta
_{b}^{(0)})^{1/2},  \label{ARXFNorm-def}
\end{equation}%
$h_{n}\in \{0,...,2s_{n}\}$ for all the $n\in \{1,...,\mathsf{N}\}$ and:%
\begin{equation}
\eta _{n}^{(k_{n})}\equiv \eta _{n}^{-}+(s_{n}-k_{n})\eta ,
\end{equation}%
define a $\mathsf{D}%
$-eigenbasis of $\mathcal{L}_{\mathsf{N}}$:%
\begin{equation}
\langle \text{\textbf{h}}|\mathsf{D}(\lambda )=d_{\text{\textbf{h}}}(\lambda
)\langle \text{\textbf{h}}|,  \label{ARXFD-L-EigenV}
\end{equation}%
where:%
\begin{equation}
d_{\text{\textbf{h}}}(\lambda )\equiv \prod_{n=1}^{\mathsf{N}}(\lambda -\eta
_{n}^{(h_{n})})\text{ \ \ \ and \ \ \textbf{h}}\equiv (h_{1},...,h_{\mathsf{N%
}}).  \label{ARXFEigenValue-D}
\end{equation}%
Moreover it holds:%
\begin{eqnarray}
\langle \text{\textbf{h}}|\mathsf{C}(\lambda ) &=&\sum_{a=1}^{\mathsf{N}%
}\prod_{b\neq a}\frac{\lambda -\eta _{b}^{(h_{b})}}{\eta _{a}^{(h_{a})}-\eta
_{b}^{(h_{b})}}d(\eta _{a}^{(h_{a}+1-\beta _{h_{a}})})\langle \text{\textbf{h%
}}|\text{T}_{a}^{+},  \label{ARXFC-SOV_D-left} \\
&&  \notag \\
\langle \text{\textbf{h}}|\mathsf{B}(\lambda ) &=&\sum_{a=1}^{\mathsf{N}%
}\prod_{b\neq a}\frac{\lambda -\eta _{b}^{(h_{b})}}{\eta _{a}^{(h_{a})}-\eta
_{b}^{(h_{b})}}a(\eta _{a}^{(h_{a}-1+\alpha _{h_{a}})})\langle \text{\textbf{%
h}}|\text{T}_{a}^{-},  \label{ARXFB-SOV_D-left}
\end{eqnarray}%
where:%
\begin{equation}
\alpha _{h_{a}}\equiv (2s_{a}+1)\delta _{h_{a},0},\text{ \ \ }\beta
_{h_{a}}\equiv (2s_{a}+1)\delta _{h_{a},2s_{a}},\text{\ \ }\langle
h_{1},...,h_{a},...,h_{\mathsf{N}}|\text{T}_{a}^{\pm }=\langle
h_{1},...,h_{a}\pm 1,...,h_{\mathsf{N}}|.
\end{equation}%
Finally, $\mathsf{A}(\lambda )$ is uniquely defined by the quantum
determinant relation.\smallskip

\textsf{II)} \underline{Right $\mathsf{D}(\lambda )$ SOV-representations} \
If $\left( \ref{ARXFE-SOV}\right) $ are verified, the states $|$\textbf{h}$\rangle \equiv |h_{1},...,h_{\mathsf{N}}\rangle $, defined by:
\begin{equation}
|\text{\textbf{h}}\rangle \equiv \frac{1}{\text{\textsc{n}}}\prod_{n=1}^{%
\mathsf{N}}\prod_{k_{n}=0}^{h_{n}-1}\frac{B(\eta _{n}^{(k_{n})})}{a(\eta
_{n}^{(k_{n})})}|0\rangle ,  \label{ARXFD-right-eigenstates}
\end{equation}%
where\ $h_{n}\in \{0,...,2s_{n}\}$ for all the $n\in \{1,...,\mathsf{N}\}$, define a $\mathsf{D}%
$-eigenbasis of $\mathcal{R}_{\mathsf{N}}$:%
\begin{equation}
\mathsf{D}(\lambda )|\text{\textbf{h}}\rangle =d_{\text{\textbf{h}}}(\lambda
)|\text{\textbf{h}}\rangle .  \label{ARXFD-R-EigenV}
\end{equation}%
Moreover, it holds:%
\begin{eqnarray}
\mathsf{C}(\lambda )|\text{\textbf{h}}\rangle &=&\sum_{a=1}^{\mathsf{N}}%
\text{T}_{a}^{-}|\text{\textbf{h}}\rangle \prod_{b\neq a}\frac{\lambda -\eta
_{b}^{(h_{b})}}{\eta _{a}^{(h_{a})}-\eta _{b}^{(h_{b})}}d(\eta
_{a}^{(h_{a})}),  \label{ARXFC-SOV_D-right} \\
&&  \notag \\
\mathsf{B}(\lambda )|\text{\textbf{h}}\rangle &=&\sum_{a=1}^{\mathsf{N}}%
\text{T}_{a}^{+}|\text{\textbf{h}}\rangle \prod_{b\neq a}\frac{\lambda -\eta
_{b}^{(h_{b})}}{\eta _{a}^{(h_{a})}-\eta _{b}^{(h_{b})}}a(\eta
_{a}^{(h_{a})}),  \label{ARXFB-SOV_D-right}
\end{eqnarray}%
where:%
\begin{equation}
\text{T}_{a}^{\pm }|h_{1},...,h_{a},...,h_{\mathsf{N}}\rangle
=|h_{1},...,h_{a}\pm 1,...,h_{\mathsf{N}}\rangle .
\end{equation}%
Finally, $\mathsf{A}(\lambda )$ is uniquely defined by the quantum
determinant relation.
\end{theorem}

\begin{proof}
The proof of the theorem is based on Yang-Baxter commutation relations and
on the fact that the left and right references states are $\mathsf{D}%
$-eigenstates:%
\begin{equation}
\langle 0|\mathsf{A}(\lambda )=a(\lambda )\langle 0|,\text{ \ \ \ }\langle 0|%
\mathsf{D}(\lambda )=d(\lambda )\langle 0|,\text{ \ \ \ }\langle 0|\mathsf{B}%
(\lambda )=\text{\b{0}},\text{ \ \ \ }\langle 0|\mathsf{C}(\lambda )\neq 
\text{\b{0}},  \label{ARXFL_ref-E}
\end{equation}

and%
\begin{equation}
\mathsf{A}(\lambda )|0\rangle =a(\lambda )|0\rangle ,\text{ \ \ \ }\mathsf{D}%
(\lambda )|0\rangle =d(\lambda )|0\rangle ,\text{ \ \ \ }\mathsf{C}(\lambda
)|0\rangle =\text{\b{0}},\text{ \ \ \ }\mathsf{B}(\lambda )|0\rangle \neq 
\text{\b{0}}.
\end{equation}
Indeed, to prove that $\left( \ref{ARXFD-left-eigenstates}\right) $ and $\left( %
\ref{ARXFD-right-eigenstates}\right) $ are left and right eigenstates of $%
\mathsf{D}(\lambda)$ as stated in $\left( \ref{ARXFD-L-EigenV}\right) $ and $%
\left( \ref{ARXFD-R-EigenV}\right) $, we have just to repeat the standard
computations in algebraic Bethe ansatz \cite{ARXFF95} as done in the proof of Theorem 3.2 of \cite{ARXFN12-0}.
\end{proof}

Note that representations of the type $\left( \ref{ARXFC-SOV_D-left}\right) $-$%
\left( \ref{ARXFB-SOV_D-left}\right) $ and $\left( \ref{ARXFC-SOV_D-right}\right) $-$%
\left( \ref{ARXFB-SOV_D-right}\right) $ for the generators of the 6-vertex
Yang-Baxter algebra are also derived from the original representations by
the change of basis associated to the factorizing $F$-matrices \cite{ARXFMaiS00}%
. These matrices give explicit representations of the Drinfel'd's twist of
quasi-triangular quasi-Hopf algebras \cite{ARXFDr1}-\cite{ARXFDr3} and their
connection with Sklyanin's quantum separation of variables was pointed out
in \cite{ARXFT99} providing the factorizing $F$-matrices for general Yangian $%
Y(sl(2))$.

\subsection{SOV-decomposition of the identity}

The following proposition holds:

\begin{proposition}
Let $\langle $\textbf{h}$|$ be the generic left $\mathsf{D}$-eigenstate and $%
|$\textbf{k}$\rangle $ be the generic right $\mathsf{D}$-eigenstate, then
the action of the covector $\langle $\textbf{h}$|$ on the vector $|$\textbf{k%
}$\rangle $ reads:%
\begin{equation}
\langle \text{\textbf{h}}|\text{\textbf{k}}\rangle =\prod_{c=1}^{\mathsf{N}%
}\delta _{h_{c},k_{c}}\prod_{1\leq b<a\leq \mathsf{N}}\frac{1}{\eta
_{a}^{(h_{a})}-\eta _{b}^{(h_{b})}}.  \label{ARXFh|k}
\end{equation}
\end{proposition}

\begin{proof}
The identity:%
\begin{equation}
(d_{\text{\textbf{h}}}(\lambda )-d_{\text{\textbf{k}}}(\lambda ))\langle 
\text{\textbf{h}}|\text{\textbf{k}}\rangle =0\text{ \ \ }\forall \lambda \in 
\mathbb{C},
\end{equation}%
obtained by computing $\langle $\textbf{h}$|\mathsf{D}(\lambda )|$\textbf{k}$%
\rangle $ implies: 
\begin{equation}
\langle \text{\textbf{h}}|\text{\textbf{k}}\rangle =0\text{ \ \ }\forall 
\text{\textbf{h}}\neq \text{\textbf{k}}\in \{0,...,2s_{1}\}\times ....\times
\{0,...,2s_{\mathsf{N}}\},  \label{ARXFortho-LR}
\end{equation}%
then we have just to compute $\langle $\textbf{h}$|$\textbf{h}$\rangle $. In
order to do so we compute $\theta _{a}\equiv \langle
h_{1},...,h_{a}-1,...,h_{\mathsf{N}}|C(\eta
_{a}^{(h_{a}-1)})|h_{1},...,h_{a},...,h_{\mathsf{N}}\rangle $, where $1\leq
h_{a}\leq 2s_{a}$ and $a\in \{1,...,\mathsf{N}\}$. In our SOV-representation
the left action of $C(\eta _{a}^{(h_{a}-1)})$ reads:%
\begin{equation}
\theta _{a}=d(\eta _{a}^{(h_{a})})\langle h_{1},...,h_{a},...,h_{\mathsf{N}%
}|h_{1},...,h_{a},...,h_{\mathsf{N}}\rangle ,\text{ being }\beta _{h_{a}-1}=0%
\text{ \ for }1\leq h_{a}\leq 2s_{a},
\end{equation}%
while the right action plus the identity $\left( \ref{ARXFortho-LR}\right) $
imply:%
\begin{equation}
\theta _{a}=\prod_{b\neq a,b=1}^{\mathsf{N}}\frac{\eta _{a}^{(h_{a}-1)}-\eta
_{b}^{(h_{b})}}{\eta _{a}^{(h_{a})}-\eta _{b}^{(h_{b})}}d(\eta
_{a}^{(h_{a})})\langle h_{1},...,h_{a}-1,...,h_{\mathsf{N}%
}|h_{1},...,h_{a}-1,...,h_{\mathsf{N}}\rangle,
\end{equation}%
and so:%
\begin{equation}
\frac{\langle h_{1},...,h_{a},...,h_{\mathsf{N}}|h_{1},...,h_{a},...,h_{%
\mathsf{N}}\rangle }{\langle h_{1},...,h_{a}-1,...,h_{\mathsf{N}%
}|h_{1},...,h_{a}-1,...,h_{\mathsf{N}}\rangle }=\prod_{b\neq a,b=1}^{\mathsf{%
N}}\frac{\eta _{a}^{(h_{a}-1)}-\eta _{b}^{(h_{b})}}{\eta _{a}^{(h_{a})}-\eta
_{b}^{(h_{b})}}.  \label{ARXFF1}
\end{equation}

This result implies:%
\begin{equation}
\frac{\langle h_{1},...,h_{\mathsf{N}}|h_{1},...,h_{\mathsf{N}}\rangle }{%
\langle 0|0\rangle /\text{\textsc{n}}^{2}}=\prod_{1\leq b<a\leq \mathsf{N}}%
\frac{\eta _{a}^{(0)}-\eta _{b}^{(0)}}{\eta _{a}^{(h_{a})}-\eta
_{b}^{(h_{b})}},  \label{ARXFF2}
\end{equation}%
then the definition $\left( \ref{ARXFNorm-def}\right) $ of the normalization 
\textsc{n}\ and the fact that:%
\begin{equation}
\langle 0|0\rangle =1,
\end{equation}%
imply $\left( \ref{ARXFh|k}\right) $.
\end{proof}

From the previous result and from the fact that $\mathsf{D}(\lambda )$ is
diagonalizable and with simple spectrum we get the following decomposition
of the identity $\mathbb{I}$:%
\begin{equation}
\mathbb{I}\equiv \sum_{h_{1}=0}^{2s_{1}}...\sum_{h_{\mathsf{N}}=0}^{2s_{%
\mathsf{N}}}\prod_{1\leq b<a\leq \mathsf{N}}(\eta _{a}^{(h_{a})}-\eta
_{b}^{(h_{b})})|h_{1},...,h_{\mathsf{N}}\rangle \langle h_{1},...,h_{\mathsf{%
N}}|.  \label{ARXFDecomp-Id}
\end{equation}

\section{$\mathsf{\bar{T}}$-spectrum characterization by SOV}

Let 
\begin{equation}
\mathbb{C}_{even}[\lambda ]_{\mathsf{N}-1}\text{ for }\mathsf{N}\text{ odd,
\ \ \ }\mathbb{C}_{odd}[\lambda ]_{\mathsf{N}-1}\text{ for }\mathsf{N}\text{
even},  \label{ARXFset-t}
\end{equation}%
be the linear spaces in the field $\mathbb{C}$ of the polynomials of degree $%
\mathsf{N}-1$ in the variable $\lambda $ even or odd as stated in the
subscript, then, the set of the eigenvalue functions $t(\lambda )$ of $%
\mathsf{\bar{T}}(\lambda )$, $\Sigma _{\mathsf{\bar{T}}}$ is contained in (%
\ref{ARXFset-t}) and moreover it holds:

\begin{theorem}
\label{ARXFC:T-eigenstates}Under the conditions $\left( \ref{ARXFE-SOV}\right) $,
the spectrum of $\bar{\mathsf{T}}(\lambda )$ is simple and $\Sigma _{{%
\mathsf{\bar{T}}}}$ coincides with the set of solutions of the discrete
system of equations:%
\begin{equation}
\det_{2s_{n}+1}D_{n}=0,\text{ \ \ }\forall n\in \{1,...,\mathsf{N}\},
\label{ARXFI-Functional-eq}
\end{equation}%
in (\ref{ARXFset-t}). Here, $D_{n}$ is the $\left( 2s_{n}+1\right) \times \left(
2s_{n}+1\right) $ tridiagonal matrix:%
\begin{equation}
D_{n}\equiv \left( 
\begin{array}{cccccc}
t(\eta _{n}^{(0)}) & -a(\eta _{n}^{(0)}) & 0\ldots &  & 0 & 0 \\ 
-d(\eta _{n}^{(1)}) & t(\eta _{n}^{(1)}) & -a(\eta _{n}^{(1)})\ldots &  & 0
& 0 \\ 
0 & {\quad }\ddots &  &  &  &  \\ 
\vdots &  & \ddots &  &  &  \\ 
\vdots &  & 0\ldots & -d(\eta _{n}^{(2s_{n}-1)}) & t(\eta _{n}^{(2s_{n}-1)})
& -a(\eta _{n}^{(2s_{n}-1)}) \\ 
0 & \ldots & 0\ldots & 0 & -d(\eta _{n}^{(2s_{n})}) & t(\eta _{n}^{(2s_{n})})%
\end{array}%
\right) .
\end{equation}

\begin{itemize}
\item[\textsf{I)}] Up to an overall normalization, the right $\mathsf{\bar{T}%
}$-eigenstate corresponding to $t(\lambda )\in \Sigma _{\mathsf{\bar{T}}}$
reads:%
\begin{equation}
|t\rangle =\sum_{h_{1}=0}^{2s_{1}}...\sum_{h_{\mathsf{N}}=0}^{2s_{\mathsf{N}%
}}\prod_{a=1}^{\mathsf{N}}Q_{t}(\eta _{a}^{(h_{a})})\prod_{1\leq b<a\leq 
\mathsf{N}}(\eta _{a}^{(h_{a})}-\eta _{b}^{(h_{b})})|h_{1},...,h_{\mathsf{N}%
}\rangle ,  \label{ARXFeigenT-r-D}
\end{equation}%
where:%
\begin{eqnarray}
Q_{t}(\eta _{a}^{(h_{a}+1)}) &=&\frac{t(\eta _{a}^{(h_{a})})}{d(\eta
_{a}^{(h_{a}+1)})}Q_{t}(\eta _{a}^{(h_{a})})-\frac{a(\eta _{a}^{(h_{a}-1)})}{%
d(\eta _{a}^{(h_{a}+1)})}Q_{t}(\eta _{a}^{(h_{a}-1)}),\text{\ \ }\forall
h_{a}\in \{1,...,2s_{a}-1\},  \label{ARXFt-Q-relation-h} \\
Q_{t}(\eta _{a}^{(1)}) &=&Q_{t}(\eta _{a}^{(0)})t(\eta _{a}^{(0)})/d(\eta
_{a}^{(1)}).  \label{ARXFt-Q-relation}
\end{eqnarray}

\item[\textsf{II)}] Up to an overall normalization, the left $\mathsf{\bar{T}%
}$-eigenstate corresponding to $t(\lambda )\in \Sigma _{\mathsf{\bar{T}}}$
reads:%
\begin{equation}
\langle t|=\sum_{h_{1}=0}^{2s_{1}}...\sum_{h_{\mathsf{N}}=0}^{2s_{\mathsf{N}%
}}\prod_{a=1}^{\mathsf{N}}\bar{Q}_{t}(\eta _{a}^{(h_{a})})\prod_{1\leq
b<a\leq \mathsf{N}}(\eta _{a}^{(h_{a})}-\eta _{b}^{(h_{b})})\langle
h_{1},...,h_{\mathsf{N}}|,  \label{ARXFeigenT-l-D}
\end{equation}%
where:%
\begin{eqnarray}
\bar{Q}_{t}(\eta _{a}^{(h_{a}+1)}) &=&\frac{t(\eta _{a}^{(h_{a})})}{a(\eta
_{a}^{(h_{a})})}\bar{Q}_{t}(\eta _{a}^{(h_{a})})-\frac{d(\eta _{a}^{(h_{a})})%
}{a(\eta _{a}^{(h_{a})})}\bar{Q}_{t}(\eta _{a}^{(h_{a}-1)}),\text{\ \ }%
\forall h_{a}\in \{1,...,2s_{a}-1\},  \label{ARXFt-Qbar-relation-h} \\
\bar{Q}_{t}(\eta _{a}^{(1)}) &=&\bar{Q}_{t}(\eta _{a}^{(0)})t(\eta
_{a}^{(0)})/a(\eta _{a}^{(0)}).  \label{ARXFt-Qbar-relation}
\end{eqnarray}
\end{itemize}
\end{theorem}

\begin{proof}
The \textit{wave-functions:}%
\begin{equation}
\Psi _{t}(\text{\textbf{h}})\equiv \langle t|\text{\textbf{h}}\rangle ,
\end{equation}%
which are the coefficients of $\langle t|$, $\mathsf{\bar{T}}$-eigenstate
corresponding to the eigenvalue $t(\lambda )\in \Sigma _{{\mathsf{\bar{T}}}}$%
, in the SOV-decomposition of the identity \rf{ARXFDecomp-Id}, satisfy the following discrete
system of $d_{\mathsf{N}}\equiv \prod_{n=1}^{\mathsf{N}}(2s_{n}+1)$
Baxter-like equations:%
\begin{equation}
t(\eta {}_{n}^{(h_{n})})\Psi _{t}(\text{\textbf{h}})\,=\,a(\eta
{}_{n}^{(h_{n})})\Psi _{t}(\mathsf{T}_{n}^{+}(\text{\textbf{h}}))+d(\eta
{}_{n}^{(h_{n})})\Psi _{t}(\mathsf{T}_{n}^{-}(\text{\textbf{h}})),
\label{ARXFSOVBax1}
\end{equation}%
for any$\,n\in \{1,...,$\textsf{$N$}$\}$ and \textbf{h}$\in \otimes _{n=1}^{%
\mathsf{N}}\{0,...,2s_{n}\}$, where:%
\begin{equation}
\mathsf{T}_{n}^{\pm }(\text{\textbf{h}})\equiv (h_{1},\dots ,h_{n}\pm
1,\dots ,h_{\mathsf{N}}).
\end{equation}%
Now observing that:%
\begin{equation}
a(\eta _{n}^{(2s_{n})})=d(\eta _{n}^{(0)})=0,
\end{equation}%
we can rephrase the previous system by the homogeneous system of equations:%
\begin{equation}
D_{n}\left( 
\begin{array}{l}
\Psi _{t}(\cdots ,h_{n}=0,\cdots ) \\ 
\Psi _{t}(\cdots ,h_{n}=1,\cdots ) \\ 
\vdots \\ 
\vdots \\ 
\Psi _{t}(\cdots ,h_{n}=2s_{n},\cdots )%
\end{array}%
\right) _{\left( 2s_{n}+1\right) \times 1}=\left( 
\begin{array}{l}
0 \\ 
\vdots \\ 
\vdots \\ 
\vdots \\ 
0%
\end{array}%
\right) _{\left( 2s_{n}+1\right) \times 1}\text{ \ \ \ }\forall n\in \{1,...,%
\mathsf{N}\}.  \label{ARXFhomo-system}
\end{equation}%
From $t(\lambda )\in \Sigma _{{\mathsf{\bar{T}}}}$ follows that the
determinants of the $\left( 2s_{n}+1\right) \times \left( 2s_{n}+1\right) $
matrix $D_{n}$ must be zero for any$\,n\in \{1,...,$\textsf{$N$}$\}$, i.e. we get
the system of equations $\left( \ref{ARXFI-Functional-eq}\right) $.\ Let $\det
\left( D_{n}\right) _{i,j}$ be the minor obtained by eliminating the row $i$
and the column $j$ from the matrix $D_{n}$, then being:%
\begin{equation}
\det \left( D_{n}\right) _{1,2s_{n}+1}=\prod_{h_{n}=1}^{2s_{n}}d(\eta
_{n}^{(h_{n})})\neq 0,  \label{ARXFRank1}
\end{equation}%
the matrices $D_{n}$ have rank $2s_{n}$ and then the solution of $\left( \ref{ARXFhomo-system}\right) $ is unique up to an overall normalization. Then for
any $t(\lambda )\in \Sigma _{{\mathsf{\bar{T}}}}$ there exist (up to
normalization) one and only one $\mathsf{\bar{T}}$-eigenstate $\langle t|$
characterized by: 
\begin{equation}
\frac{\Psi _{t}(h_{1},...,h_{n}+1,...,h_{1})}{\Psi
_{t}(h_{1},...,h_{n},...,h_{1})}=\frac{\bar{Q}_{t}(\eta _{n}^{(h_{n}+1)})}{%
\bar{Q}_{t}(\eta _{n}^{(h_{n})})},
\end{equation}%
for any$\,n\in \{1,...,$\textsf{$N$}$\}$, $h_{n}\in \{0,...,2s_{n}-1\}$ and $%
h_{m\neq n}\in \{0,...,2s_{m}\}$ as the equations $\left( \ref{ARXFt-Qbar-relation-h}\right) $-$\left( \ref{ARXFt-Qbar-relation}\right) $ fix
uniquely the ratios on the r.h.s. of the above equations; this proves the
simplicity of the $\mathsf{\bar{T}}$-spectrum.

Vice versa, let $t(\lambda )$ be a solution of (\ref{ARXFI-Functional-eq}) in (\ref{ARXFset-t}), then the state $|\,t\,\rangle $
constructed by $\left( \ref{ARXFeigenT-l-D}\right) $-$\left( \ref{ARXFt-Qbar-relation}\right) $ satisfies:
\begin{equation}
\langle t|\mathsf{\bar{T}}(\eta _{n}^{(h_{n})})|h_{1},...,h_{\mathsf{N}%
}\rangle =t(\eta _{n}^{(h_{n})})\langle t|h_{1},...,h_{\mathsf{N}}\rangle 
\text{ \ }\forall n\in \{1,...,\mathsf{N}\}
\end{equation}%
for any $\mathsf{D}$-eigenstate $\langle h_{1},...,h_{\mathsf{N}}|$. This
implies:%
\begin{equation}
\langle h_{1},...,h_{\mathsf{N}}|\mathsf{\bar{T}}(\lambda )|\,t\,\rangle
=t(\lambda )\langle h_{1},...,h_{\mathsf{N}}|\,t\,\rangle ,
\end{equation}%
that is $t(\lambda )\in \Sigma _{{\mathsf{\bar{T}}}}$\ and $|\,t\,\rangle $
is the corresponding $\mathsf{\bar{T}}$-eigenstate as $\mathsf{\bar{T}}%
(\lambda )$ is a polynomials of degree $\mathsf{N}-1$ in $\lambda $, even or
odd for $\mathsf{N}$ odd or even.
\end{proof}

Note that the Baxter Q-operator construction given in \cite{ARXFYB95} can be adapted in
particular to the representation here considered and it represents an
interesting issue as allows to reformulate by functional equations the SOV
characterization of $\mathsf{\bar{T}}$-spectrum.

\section{$\mathsf{\bar{T}}$-decomposition of the identity}\label{ARXFSP}

As previously introduced in \cite{ARXFGMN12-SG,ARXFN12-0} and \cite{ARXFN12-2}, we
recall the definition of separate covector $\langle \alpha |$ and vector $|\beta
\rangle $ in the SOV representations:%
\begin{align}
\langle \alpha |& =\sum_{h_{1}=0}^{2s_{1}}\cdots \sum_{h_{\mathsf{N}%
}=0}^{2s_{\mathsf{N}}}\prod_{a=1}^{\mathsf{N}}\alpha _{a}(\eta
_{a}^{(h_{a})})\prod_{1\leq b<a\leq \mathsf{N}}(\eta _{a}^{(h_{a})}-\eta
_{b}^{(h_{b})})\langle h_{1},...,h_{\mathsf{N}}|, \\
|\beta \rangle & =\sum_{h_{1}=0}^{2s_{1}}\cdots \sum_{h_{\mathsf{N}}=0}^{2s_{%
\mathsf{N}}}\prod_{a=1}^{\mathsf{N}}\beta _{a}(\eta
_{a}^{(h_{a})})\prod_{1\leq b<a\leq \mathsf{N}}(\eta _{a}^{(h_{a})}-\eta
_{b}^{(h_{b})})|h_{1},...,h_{\mathsf{N}}\rangle ,
\end{align}%
and we show that also for the model under consideration the following
results hold:

\begin{proposition}
The generic separate covector $\langle \alpha |$ acts on the generic separate
vector $|\beta \rangle $ according to the following formula:%
\begin{equation}
\langle \alpha |\beta \rangle =\det_{\mathsf{N}}||\mathcal{M}_{a,b}^{\left(
\alpha ,\beta \right) }||\text{ \ \ with \ }\mathcal{M}_{a,b}^{\left( \alpha
,\beta \right) }\equiv \sum_{h=0}^{2s_{a}}\alpha _{a}(\eta _{a}^{(h)})\beta
_{a}(\eta _{a}^{(h)})\left( \eta _{a}^{(h)}\right) ^{b-1}.  \label{ARXFSP-det}
\end{equation}
\end{proposition}

\begin{proof}
The proof is a consequence of formula $\left( \ref{ARXFh|k}\right) $
which implies:%
\begin{equation}
\langle \alpha |\beta \rangle =\sum_{h_{1}=0}^{2s_{1}}\cdots \sum_{h_{%
\mathsf{N}}=0}^{2s_{\mathsf{N}}}V(\eta _{1}^{(h_{1})},...,\eta _{\mathsf{N}%
}^{(h_{\mathsf{N}})})\prod_{a=1}^{\mathsf{N}}\alpha _{a}(\eta
_{a}^{(h_{a})})\beta _{a}(\eta _{a}^{(h_{a})}),
\end{equation}%
where $V(x_{1},...,x_{\mathsf{N}})\equiv \prod_{1\leq b<a\leq \mathsf{N}%
}(x_{a}-x_{b})$ is the Vandermonde determinant, and of the multilinearity of
the determinant.
\end{proof}

The scalar product introduced in Section \ref{ARXFYB-LR} allows to characterize $%
\langle \alpha |$ as the dual of a vector $\left( \langle \alpha |\right)
^{\dag }\in \mathcal{R}_{\mathsf{N}}$, then the previous formula represents also
the scalar product of two states in $\mathcal{R}_{\mathsf{N}}$. Indeed, from
the Hermitian conjugation properties $\left( \ref{ARXFHerm-1}\right) $ and $%
\left( \ref{ARXFHerm-2}\right) $, it follows that $\left( \langle \alpha
|\right) ^{\dag }$ is a separate vector w.r.t. the $\mathsf{A}$%
-decomposition of the identity. Finally, let us remark that for the transfer
matrix the orthogonality of eigenstates corresponding to different
eigenvalues can be proven directly by using the previous scalar product
formula:

\begin{corollary}
Let us take $t(\lambda )$ and $t^{\prime }(\lambda )\in \Sigma _{{\mathsf{%
\bar{T}}}}$ and let $\langle t|$ and $|t^{\prime }\rangle $ be the
corresponding $\mathsf{\bar{T}}$-eigenstates characterized in Theorem $\ref{ARXFC:T-eigenstates}$, then for $t(\lambda )\neq t^{\prime }(\lambda )$ the $%
\mathsf{N}\times \mathsf{N}$ matrix $||\mathcal{M}_{a,b}^{\left( t,t^{\prime
}\right) }||$ has rank equal or smaller than $\mathsf{N}-1$. Indeed, the non-zero $\mathsf{N}\times 1$ vector V$^{\left( t,t^{\prime }\right) }$ of components:%
\begin{equation}
\text{V}_{b}^{\left( t,t^{\prime }\right) }\equiv c_{b}^{\prime }-c_{b}\text{%
\ \ \ }\forall b\in \{1,...,\mathsf{N}\},  \label{ARXFV-vector}
\end{equation}%
where:%
\begin{equation}
t(\lambda )=\sum_{b=1}^{\mathsf{N}}c_{b}\lambda ^{b-1},\text{ \ \ }t^{\prime
}(\lambda )=\sum_{b=1}^{\mathsf{N}}c_{b}^{\prime }\lambda ^{b-1},
\label{ARXFt-t'-decomp}
\end{equation}%
is an eigenvector of $||\mathcal{M}_{a,b}^{\left( t,t^{\prime }\right) }||$
corresponding to the eigenvalue zero.
\end{corollary}

\begin{proof}
By the definition $\left( \ref{ARXFV-vector}\right) $, $\left( \ref{ARXFt-t'-decomp}%
\right) $ and the definition of $\mathcal{M}_{a,b}^{\left( t,t^{\prime
}\right) }$ in $\left( \ref{ARXFSP-det}\right) $, it holds:%
\begin{equation}
\sum_{b=1}^{\mathsf{N}}\mathcal{M}_{a,b}^{\left( t,t^{\prime }\right) }\text{%
V}_{b}^{\left( t,t^{\prime }\right) }=\sum_{h=0}^{2s_{a}}Q_{t^{\prime
}}(\eta _{a}^{(h)})\bar{Q}_{t}(\eta _{a}^{(h)})(t^{\prime }(\eta
_{a}^{(h)})-t(\eta _{a}^{(h)})),  \label{ARXFzero-eigenvector-1}
\end{equation}%
then by using the Baxter relations $(\ref{ARXFt-Q-relation-h})$-$(\ref{ARXFt-Q-relation})$ and $(\ref{ARXFt-Qbar-relation-h})$-$(\ref{ARXFt-Qbar-relation})$
we can write:%
\begin{align}
Q_{t^{\prime }}(\eta _{a}^{(h)})\bar{Q}_{t}(\eta _{a}^{(h)})(t^{\prime
}(\eta _{a}^{(h)})-t(\eta _{a}^{(h)}))& =(d(\eta _{a}^{(h+1-\beta
_{h})})Q_{t^{\prime }}(\eta _{a}^{(h+1)})+a(\eta _{a}^{(h-1+\alpha
_{h})})Q_{t^{\prime }}(\eta _{a}^{(h-1)}))\bar{Q}_{t}(\eta _{a}^{(h)}) 
\notag \\
& -(a(\eta _{a}^{(h)})\bar{Q}_{t}(\eta _{a}^{(h+1)})+d(\eta _{a}^{(h)})\bar{Q%
}_{t}(\eta _{a}^{(h-1)}))Q_{t^{\prime }}(\eta _{a}^{(h)}),
\end{align}%
and by substituting them in (\ref{ARXFzero-eigenvector-1}) we get our result:%
\begin{equation}
\sum_{b=1}^{\mathsf{N}}\mathcal{M} _{a,b}^{\left( t,t^{\prime }\right) }\text{V}%
_{b}^{\left( t,t^{\prime }\right) }=0\text{ \ \ \ \ }\forall a\in \{1,...,%
\mathsf{N}\}.  \label{ARXFzero-eigenvector}
\end{equation}
\end{proof}

The normality of $\mathsf{\bar{T}}(\lambda )$ and the simplicity of its
spectrum imply the following decomposition of the identity:%
\begin{equation}
\mathbb{I=}\sum_{t(\lambda )\in \Sigma _{\mathsf{\bar{T}}}}\frac{|t\rangle
\langle t|}{\langle t|t\rangle }\text{ \ \ with }\langle t|t\rangle =\det_{%
\mathsf{N}}||\mathcal{M} _{a,b}^{\left( t,t\right) }||,
\label{ARXFT-Id-decomp}
\end{equation}%
where for the $\mathsf{\bar{T}}$-eigenstates $\langle t|$ and $|t\rangle $ we
are using the characterization given in Theorem $\ref{ARXFC:T-eigenstates}$.

\section{\label{ARXFRecLSO}Reconstruction of local spin generators}

\subsection{Antiperiodic solution of the quantum inverse problem}

\noindent Here we show how to write the solution of the quantum inverse
problem in the case of the antiperiodic XXX spin $\{s_{1},...,s_{\mathsf{N}%
}\}$-chain. The solution here presented is a simple consequence of the
results derived in \cite{ARXFMaiT00, ARXFCM07} for the periodic chain.

\begin{proposition}
In the XXX spin $\{s_{1},...,s_{\mathsf{N}}\}$-chain, the generic local
operator $X_{n}\in $End$(V_{n}^{(s_{n})})$ at any quantum site $n\in
\{1,..., \mathsf{N}\}$ admits the following reconstruction: 
\begin{equation}
X_{n}=\prod_{k=1}^{n-1}\mathsf{\bar{T}}^{(s_{k})}(\eta _{k})\text{$tr$}%
_{0_{n}}\left( X_{0_{n}}\mathsf{\bar{M}}_{0_{n}}^{(s_{n})}(\eta _{n})\right)
\prod_{k=1}^{n}\left( \mathsf{\bar{T}}^{(s_{k})}(\eta _{k})\right) ^{-1},
\label{ARXFInR1}
\end{equation}%
where the auxiliary space $0_{n}$ is isomorphic to R$_{n}^{(s_{n})}$.\
Moreover, for $X_{n}\equiv S_{n}^{\alpha }$ and $\alpha \equiv \pm ,z$, the
following decompositions hold: 
\begin{align}
\text{$tr$}_{0_{n}}\left( S_{0_{n}}^{\alpha }\mathsf{\bar{M}}%
_{0_{n}}^{(s_{n})}(\eta _{n})\right) & =\sum_{k=1}^{2s_{n}}\mathsf{\bar{T}}%
^{(s_{n}-\frac{k}{2})}\left( \eta _{n}+\frac{k\eta }{2}\right) \text{$tr$}%
_{0}\left( S_{0}^{\alpha }\mathsf{\bar{M}}_{0}^{(1/2)}(\eta
_{n}^{-}+(k-s_{n})\eta )\right)  \notag \\
& \times \mathsf{\bar{T}}^{(\frac{k-1}{2})}\left( \eta _{n}^{-}+\frac{%
(k-2s_{n})\eta }{2}\right) ,  \label{ARXFR1} \\
& =\sum_{k=1}^{2s_{n}}\mathsf{\bar{T}}^{(\frac{k-1}{2})}\left( \eta _{n}^{-}+%
\frac{(k-2s_{n})\eta }{2}\right) \text{$tr$}_{0}\left( S_{0}^{\alpha }%
\mathsf{\bar{M}}_{0}^{(1/2)}(\eta _{n}^{-}+(k-s_{n})\eta )\right)  \notag \\
& \times \mathsf{\bar{T}}^{(s_{n}-\frac{k}{2})}\left( \eta _{n}+\frac{k\eta 
}{2}\right) \!\!,  \label{ARXFR2}
\end{align}%
in terms of the fused transfer matrix and the matrix elements of the basic $%
\mathsf{\bar{M}}_{0}^{(1/2)}(\lambda )$ monodromy matrix.
\end{proposition}

\begin{proof}
From Proposition 1 of the article \cite{ARXFMaiT00}, it holds:%
\begin{equation}
X_{n}=\prod_{k=1}^{n-1}\mathsf{T}^{(s_{k})}(\eta _{k})\text{$tr$}%
_{0_{n}}\left( X_{0_{n}}\mathsf{M}_{0_{n}}^{(s_{n})}(\eta _{n})\right)
\prod_{k=1}^{n}\left( \mathsf{T}^{(s_{k})}(\eta _{k})\right) ^{-1},
\label{ARXFR-P-1}
\end{equation}%
then it holds:
\begin{equation}
\Sigma _{n}^{(x)}=\prod_{k=1}^{n-1}\mathsf{T}^{(s_{k})}(\eta _{k})\mathsf{%
\bar{T}}^{(s_{n})}(\eta _{n})\prod_{k=1}^{n}\left( \mathsf{T}^{(s_{k})}(\eta
_{k})\right) ^{-1}.  \label{ARXFs^x_n}
\end{equation}%
So, we can use $(\ref{ARXFs^x_n})$ to write:%
\begin{eqnarray}
\prod_{b=1}^{c}\Sigma _{b}^{(x)} &=&\prod_{b=1}^{c}\mathsf{\bar{T}}%
^{(s_{b})}(\eta _{b})\prod_{b=1}^{c}\left( \mathsf{T}^{(s_{b})}(\eta
_{b})\right) ^{-1}  \label{ARXFP-s^x_n} \\
&=&\prod_{b=1}^{c}\mathsf{T}^{(s_{b})}(\eta _{b})\prod_{b=1}^{c}\left( 
\mathsf{\bar{T}}^{(s_{b})}(\eta _{b})\right) ^{-1},  \label{ARXFP-s^x_n-2}
\end{eqnarray}%
where the second equality follows from:%
\begin{equation}
\prod_{b=1}^{c}\Sigma _{b}^{(x)}=\prod_{b=1}^{c}\left( \Sigma
_{b}^{(x)}\right) ^{-1}
\end{equation}%
being:%
\begin{equation}
\Sigma _{b}^{(x)}\Sigma _{b}^{(x)}=\mathbb{I}_{b}\text{ \ \ where }\mathbb{I}_{b}\text{ is the
identity matrix in }\text{R}_{b}^{(s_{b})}.
\end{equation}%
The result $(\ref{ARXFInR1})$ is derived computing:%
\begin{equation}
X_{n}=\prod_{b=1}^{n-1}\Sigma _{b}^{(x)}\tilde{X}_{n}\prod_{b=1}^{n}\Sigma
_{b}^{(x)}\text{ \ \ with \ \ }\tilde{X}_{n}=X_{n}\Sigma _{n}^{(x)}
\end{equation}%
and using the reconstruction $(\ref{ARXFR-P-1})$ for $\tilde{X}_{n}$ and the reconstruction in $(\ref{ARXFP-s^x_n})$ for the first product of $\Sigma _{b}^{(x)}$
while the reconstruction in $(\ref{ARXFP-s^x_n-2})$ for the second product of $\Sigma _{b}^{(x)}$. The formula $(\ref{ARXFR1})$ can be proven following step by step the proof
given in \cite{ARXFMaiT00} only changing the periodic objects with the
antiperiodic ones. In the same way we can prove formula $(\ref{ARXFR2})$
following step by step the proof given in \cite{ARXFCM07}.
\end{proof}

\section{\label{ARXFFF}Form factors of local operators}

\begin{proposition}
\label{ARXFFF-Prop1}Let $\langle t| $ and $|t^{\prime }\rangle $ be two
eigenstates of the transfer matrix $\mathsf{\bar{T}}(\lambda )$, then it
holds:%
\begin{equation}
\langle t|S_{n}^{-}|t^{\prime }\rangle =\frac{\prod_{h=1}^{n-1}t^{(s_{h})}(%
\eta _{h})}{\prod_{h=1}^{n}t^{\prime (s_{n})}(\eta _{h})}\det_{\mathsf{N}%
+1}(||\mathcal{S}_{a,b}^{\left( -,t,t^{\prime }\right) }||)  \label{ARXFFF-S+/-}
\end{equation}%
where $||\mathcal{S}_{a,b}^{\left( -,t,t^{\prime }\right) }||$ is the $($$%
\mathsf{N}+1)\times (\mathsf{N}+1)$ matrix:%
\begin{eqnarray}
\mathcal{S}_{a,b}^{\left( -,t,t^{\prime }\right) } &\equiv &\mathcal{M}
_{a,b}^{\left( t,t^{\prime }\right) }\text{ \ for \ }a\in \{1,...,\mathsf{N}%
\},\text{\ }b\in \{1,...,\mathsf{N}+1\}, \\
\mathcal{S}_{\mathsf{N}+1,b}^{\left( -,t,t^{\prime }\right) } &\equiv
&\sum_{k=1}^{2s_{n}}\frac{t^{(s_{n}-\frac{k}{2})}\left( \bar{\eta}%
_{n}^{\left( s_{n}+(k+1)/2\right) }\right) }{t^{\prime (\frac{k-1}{2}%
)}\left( \bar{\eta}_{n}^{(k/2)}\right) }\left( \bar{\eta}_{n}^{(k)}\right)
^{b-1},\text{\ \ \ \ }b\in \{1,...,\mathsf{N}+1\},
\end{eqnarray}%
where we have used the notation:%
\begin{equation}
\bar{\eta}_{n}^{(k_{n})}\equiv \eta _{n}+(k_{n}-(s_{n}+1/2))\eta .
\end{equation}
\end{proposition}

\begin{proof}
We can compute the action of $S_{n}^{-}$, by using the reconstruction:%
\begin{equation}
S_{n}^{-}=\prod_{k=1}^{n-1}\mathsf{\bar{T}}^{(s_{k})}(\eta
_{k})\sum_{k=1}^{2s_{n}}\mathsf{\bar{T}}^{(s_{n}-\frac{k}{2})}\left( \bar{%
\eta}_{n}^{\left( s_{n}+(k+1)/2\right) }\right) D(\bar{\eta}_{n}^{(k)})%
\mathsf{\bar{T}}^{(\frac{k-1}{2})}\left( \bar{\eta}_{n}^{(k/2)}\right)
\prod_{k=1}^{n}\left( \mathsf{\bar{T}}^{(s_{k})}(\eta _{k})\right) ^{-1},
\end{equation}%
so it holds:%
\begin{equation}
\langle t|S_{n}^{-}|t^{\prime }\rangle =\frac{\prod_{h=1}^{n-1}t^{(s_{h})}(%
\eta _{h})}{\prod_{h=1}^{n}t^{\prime (s_{n})}(\eta _{h})}\sum_{k=1}^{2s_{n}}%
\frac{t^{(s_{n}-\frac{k}{2})}\left( \bar{\eta}_{n}^{\left(
s_{n}+(k+1)/2\right) }\right) }{t^{\prime (\frac{k-1}{2})}\left( \bar{\eta}%
_{n}^{(k/2)}\right) }\langle t|D(\bar{\eta}_{n}^{(k)})|t^{\prime }\rangle .
\end{equation}%
Now from the right SOV representation, we have:%
\begin{eqnarray}
D(\bar{\eta}_{n}^{(k)})|t^{\prime }\rangle &=&\sum_{h_{1}=0}^{2s_{1}}\cdots
\sum_{h_{\mathsf{N}}=0}^{2s_{\mathsf{N}}}\prod_{a=1}^{\mathsf{N}}(\bar{\eta}%
_{n}^{(k)}-\eta _{a}^{(h_{a})})Q_{t^{\prime }}(\eta _{a}^{(h_{a})})\prod_{1\leq b<a\leq \mathsf{N}}(\eta _{a}^{(h_{a})}-\eta
_{b}^{(h_{b})})|h_{1},...,h_{\mathsf{N}}\rangle ,
\end{eqnarray}%
and we can rewrite the coefficient as:%
\begin{equation}
\prod_{a=1}^{\mathsf{N}}Q_{t^{\prime }}(\eta _{a}^{(h_{a})})V(\eta
_{1}^{(h_{1})},....,\eta _{\mathsf{N}}^{(h_{\mathsf{N}})},\bar{\eta}%
_{n}^{(k)}),
\end{equation}%
where $V()$ is the determinant of the $(\mathsf{N}+1)\times (\mathsf{N}+1)$
Vandermonde matrix $\left\Vert V_{i,j}\right\Vert $ defined by:%
\begin{equation}
V_{i,j}\equiv \left( \eta _{i}^{(h_{i})}\right) ^{j-1}\text{ \ \ }\forall
i\in \{1,...,\mathsf{N}\},\text{ }V_{\mathsf{N}+1,j}\equiv \left( \bar{\eta}%
_{n}^{(k)}\right) ^{j-1}\text{ \ }\forall j\in \{1,...,\mathsf{N}+1\}.
\end{equation}%
Now taking the scalar product and resumming we get the result.
\end{proof}

\begin{proposition}
\label{ARXFFF-Prop2}Let $\langle t| $ and $|t^{\prime }\rangle $ be two
eigenstates of the transfer matrix $\mathsf{\bar{T}}(\lambda )$, then it
holds:%
\begin{equation}
\langle t|S_{n}^{z}|t^{\prime }\rangle =\det_{\mathsf{N}+1}(||\mathcal{S}%
_{a,b}^{\left( z,t,t^{\prime }\right) }||)
\end{equation}%
where $||\mathcal{S}_{a,b}^{\left( z,t,t^{\prime }\right) }||$ is the $(%
\mathsf{N}+1)\times (\mathsf{N}+1)$ matrix:%
\begin{align}
\mathcal{S}_{a,b}^{\left( z,t,t^{\prime }\right) }& \equiv\mathcal{M}
_{a,b}^{\left( t,t^{\prime }\right) }\text{ \ \ \ \ \ \ \ \ \ \ \ \ \ \ \ \
\ \ \ \ \ \ \ \ \ \ \ \ \ \ \ \ \ \ \ \ \ \ \ \ \ \ \ \ \ \ for \ }a\in
\{1,...,\mathsf{N}\},\text{ \ \ \ }b\in \{1,...,\mathsf{N}\} \\
\mathcal{S}_{\mathsf{N}+1,b}^{\left( z,t,t^{\prime }\right) }& \equiv \frac{%
\prod_{h=1}^{n-1}t^{(s_{h})}(\eta _{h})}{\prod_{h=1}^{n}t^{\prime
(s_{n})}(\eta _{h})}\sum_{h=0}^{2s_{n}-1}\frac{t^{(h/2)}\left( \eta
_{n}^{\left( (h+1)/2\right) }\right) }{t^{\prime (s_{n}-\frac{(h+1)}{2}%
)}\left( \eta _{n}^{(h/2+s_{n})}\right) }\left( \eta _{n}^{(h)}\right) ^{b-1}%
\text{ \ \ \ for }b\in \{1,...,\mathsf{N}\} \\
\mathcal{S}_{a,\mathsf{N}+1}^{\left( z,t,t^{\prime }\right) }& \equiv
\sum_{h_{a}=0}^{2s_{a}}Q_{t^{\prime }}(\eta _{a}^{(h_{a})})\bar{Q}_{t}(\eta
_{a}^{(h_{a}-1)})d(\eta _{a}^{(h_{a})})\text{ \ for \ }a\in \{1,...,\mathsf{N%
}\}, \\
\mathcal{S}_{\mathsf{N}+1,\mathsf{N}+1}^{\left( z,t,t^{\prime }\right) }&
\equiv -s_{n}.
\end{align}
\end{proposition}

\begin{proof}
We can compute the action of $S_{n}^{z}$, by using the reconstruction:%
\begin{equation}
S_{n}^{z}=\prod_{k=1}^{n-1}\mathsf{\bar{T}}^{(s_{k})}(\eta
_{k})\sum_{h=0}^{2s_{n}-1}\mathsf{\bar{T}}^{(h/2)}(\eta _{n}^{(\frac{h+1}{2}%
)})\left( \frac{C(\eta _{n}^{(h)})-B(\eta _{n}^{(h)})}{2}\right) \mathsf{%
\bar{T}}^{(s_{n}-\frac{(h+1)}{2})}\left( \eta _{n}^{(h/2+s_{n})}\right)
\prod_{k=1}^{n}\left( \mathsf{\bar{T}}^{(s_{k})}(\eta _{k})\right) ^{-1},
\end{equation}%
so it holds:%
\begin{equation}
\langle t|S_{n}^{z}|t^{\prime }\rangle =\frac{\prod_{h=1}^{n-1}t^{(s_{h})}(%
\eta _{h})}{\prod_{h=1}^{n}t^{\prime (s_{n})}(\eta _{h})}%
\sum_{h=0}^{2s_{n}-1}\frac{t^{(h/2)}\left( \eta _{n}^{\left( (h+1)/2\right)
}\right) }{t^{\prime (s_{n}-\frac{(h+1)}{2})}\left( \eta
_{n}^{(h/2+s_{n})}\right) }\langle t|C(\eta _{n}^{(h)})|t^{\prime }\rangle
-s_{n}\langle t|t^{\prime }\rangle .
\end{equation}%
Here we have used the property:%
\begin{equation}
\sum_{h=0}^{2s_{n}-1}\mathsf{\bar{T}}^{(h)}\left( \eta _{n}^{(\frac{h+1}{2}%
)}\right) \mathsf{\bar{T}}\left( \eta _{n}^{(h)}\right) \mathsf{\bar{T}}%
^{(s_{n}-\frac{(h+1)}{2})}\left( \eta _{n}^{(h/2+s_{n})}\right) =2s_{n}%
\mathsf{\bar{T}}^{(s_{n})}\left( \eta _{n}\right) ,
\end{equation}%
which is proven showing that it holds for the eigenvalues following step by
step the proof given for the formula (A.3) in \cite{ARXFCM07}. Now from the
right SOV representation of $C(\eta _{n})$, we have:%
\begin{eqnarray}
C(\eta _{n}^{(h)})|t^{\prime }\rangle &=&\sum_{a=1}^{\mathsf{N}%
}\sum_{h_{1}=0}^{2s_{1}}\cdots \sum_{h_{\mathsf{N}}=0}^{2s_{\mathsf{N}%
}}\prod_{b=1}^{\mathsf{N}}Q_{t^{\prime }}(\eta _{b}^{(h_{b})})\prod_{b\neq
a,b=1}^{\mathsf{N}}\left[ \frac{\eta _{n}^{(h)}-\eta _{b}^{(h_{b})}}{\eta
_{a}^{(h_{a})}-\eta _{b}^{(h_{b})}}d(\eta _{a}^{(h_{a})})\right]  \notag \\
&&\times \prod_{1\leq b<a\leq \mathsf{N}}(\eta _{a}^{(h_{a})}-\eta
_{b}^{(h_{b})})|h_{1},...,h_{a}-1,...,h_{\mathsf{N}}\rangle ,
\end{eqnarray}%
and so we can write:%
\begin{eqnarray}
\langle t|C(\eta _{n}^{(h)})|t^{\prime }\rangle &=&\sum_{a=1}^{\mathsf{N}%
}(-1)^{\mathsf{N}+1+a}\sum_{h_{1}=0}^{2s_{1}}\cdots \sum_{h_{\mathsf{N}%
}=0}^{2s_{\mathsf{N}}}\widehat{V}_{a}(\eta _{1}^{(h_{1})},....,\eta _{%
\mathsf{N}}^{(h_{\mathsf{N}})},\eta _{n}^{(h)})  \notag \\
&&\times Q_{t^{\prime }}(\eta _{a}^{(h_{a})})\bar{Q}_{t}(\eta
_{a}^{(h_{a}-1)})d(\eta _{a}^{(h_{a})})\prod_{b\neq a,b=1}^{\mathsf{N}%
}Q_{t^{\prime }}(\eta _{b}^{(h_{b})})\bar{Q}_{t}(\eta _{b}^{(h_{b})}),
\end{eqnarray}%
where $(-1)^{\mathsf{N}+1+a}\widehat{V}_{a}(\eta _{1}^{(h_{1})},....,\eta _{%
\mathsf{N}}^{(h_{\mathsf{N}})},\eta _{n}^{(h)})$ is the cofactor $\left( a,%
\mathsf{N}+1\right) $ of the $(\mathsf{N}+1)\times (\mathsf{N}+1)$
Vandermonde matrix:%
\begin{equation}
V_{i,j}\equiv \left( \eta _{i}^{(h_{i})}\right) ^{j-1}\text{ \ \ }\forall
i\in \{1,...,\mathsf{N}\},\text{ }V_{\mathsf{N}+1,j}\equiv \left( \eta
_{n}^{(h)}\right) ^{j-1}\text{ \ }\forall j\in \{1,...,\mathsf{N}+1\}.
\end{equation}

It is trivial to remark that the above sum is the develop of the determinant
of a $(\mathsf{N}+1)\times (\mathsf{N}+1)$ matrix and then:%
\begin{equation}
\frac{\prod_{h=1}^{n-1}t^{(s_{h})}(\eta _{h})}{\prod_{h=1}^{n}t^{\prime
(s_{n})}(\eta _{h})}\sum_{h=0}^{2s_{n}-1}\frac{t^{(h/2)}\left( \eta
_{n}^{\left( (h+1)/2\right) }\right) }{t^{\prime (s_{n}-\frac{(h+1)}{2}%
)}\left( \eta _{n}^{(h/2+s_{n})}\right) }\langle t|C(\eta
_{n}^{(h)})|t^{\prime }\rangle =\det_{\mathsf{N}+1}||\text{\textsc{s}}%
_{a,b}^{\left( t,t^{\prime }\right) }||
\end{equation}%
with:%
\begin{align}
\text{\textsc{s}}_{a,b}^{\left( t,t^{\prime }\right) }& \equiv\mathcal{M}
_{a,b}^{\left( t,t^{\prime }\right) }\text{ \ \ \ \ \ \ \ \ \ \ for \ }a\in
\{1,...,\mathsf{N}\},\text{ \ \ \ }b\in \{1,...,\mathsf{N}\} \\
\text{\textsc{s}}_{\mathsf{N}+1,b}^{\left( t,t^{\prime }\right) }& \equiv 
\frac{\prod_{h=1}^{n-1}t^{(s_{h})}(\eta _{h})}{\prod_{h=1}^{n}t^{\prime
(s_{n})}(\eta _{h})}\sum_{h=0}^{2s_{n}-1}\frac{t^{(h/2)}\left( \eta
_{n}^{\left( (h+1)/2\right) }\right) }{t^{\prime (s_{n}-\frac{(h+1)}{2}%
)}\left( \eta _{n}^{(h/2+s_{n})}\right) }\left( \eta _{n}^{(h)}\right) ^{b-1}%
\text{ \ \ \ for }b\in \{1,...,\mathsf{N}\} \\
\text{\textsc{s}}_{a,\mathsf{N}+1}^{\left( t,t^{\prime }\right) }& \equiv
\sum_{h_{a}=0}^{2s_{a}}Q_{t^{\prime }}(\eta _{a}^{(h_{a})})\bar{Q}_{t}(\eta
_{a}^{(h_{a}-1)})d(\eta _{a}^{(h_{a})})\text{ \ for \ }a\in \{1,...,\mathsf{N%
}\}, \\
\text{\textsc{s}}_{\mathsf{N}+1,\mathsf{N}+1}^{\left( t,t^{\prime }\right)
}& \equiv 0.
\end{align}%
Then if we rewrite:%
\begin{equation}
\langle t|t^{\prime }\rangle =\det_{\mathsf{N}+1}||\text{\textsc{n}}%
_{a,b}^{\left( t,t^{\prime }\right) }||,
\end{equation}%
where we have defined the matrix:%
\begin{align}
\text{\textsc{n}}_{a,b}^{\left( t,t^{\prime }\right) }& \equiv \mathcal{M}
_{a,b}^{\left( t,t^{\prime }\right) }\text{ \ \ \ \ \ \ \ \ \ \ for \ }a\in
\{1,...,\mathsf{N}\},\text{ \ \ \ }b\in \{1,...,\mathsf{N}\} \\
\text{\textsc{n}}_{\mathsf{N}+1,b}^{\left( t,t^{\prime }\right) }& \equiv 
\frac{\prod_{h=1}^{n-1}t^{(s_{h})}(\eta _{h})}{\prod_{h=1}^{n}t^{\prime
(s_{n})}(\eta _{h})}\sum_{h=0}^{2s_{n}-1}\frac{t^{(h/2)}\left( \eta
_{n}^{\left( (h+1)/2\right) }\right) }{t^{\prime (s_{n}-\frac{(h+1)}{2}%
)}\left( \eta _{n}^{(h/2+s_{n})}\right) }\left( \eta _{n}^{(k)}\right) ^{b-1}%
\text{ \ \ \ for }b\in \{1,...,\mathsf{N}\} \\
\text{\textsc{n}}_{a,\mathsf{N}+1}^{\left( t,t^{\prime }\right) }& \equiv 0%
\text{ \ for \ }a\in \{1,...,\mathsf{N}\}, \\
\text{\textsc{n}}_{\mathsf{N}+1,\mathsf{N}+1}^{\left( t,t^{\prime }\right)
}& \equiv 1.
\end{align}%
we finally get our determinant formula.
\end{proof}

\section{Conclusion and outlook}

For the higher spin representations of the rational 6-vertex
Yang-Baxter algebra we have completely characterized the transfer matrix
spectrum by separation of variables and proven its simplicity. Moreover, we have provided the
SOV-reconstruction of all local operators and determinant
formulae for the scalar product of separate states. Finally, we were able to compute the form
factors of the local spin operators getting one determinant formulae similar
to those of the scalar products.

The relevance of these findings is clear as it represents the first
fundamental step to characterize the dynamics of this quantum model. Indeed,
the decomposition of the identity $\left( \ref{ARXFT-Id-decomp}\right) $ allows
to expand any m-point function:%
\begin{equation}
\frac{\langle t|X_{r_{1}}\cdots X_{r_{\text{g}+1}}|t\rangle }{\langle
t|t\rangle }=\sum_{t_{1}(\lambda ),...,t_{\text{g}}(\lambda )\in \Sigma _{%
\mathsf{\bar{T}}}}\frac{\langle t|X_{r_{1}}|t_{1}\rangle \langle t_{\text{g}%
}|X_{r_{\text{g}}}|t\rangle \prod_{j=2}^{\text{g}}\langle
t_{j-1}|X_{r_{j}}|t_{j}\rangle }{\langle t|t\rangle \prod_{j=1}^{\text{g}%
}\langle t_{j}|t_{j}\rangle }\text{\ \ with }X_{n}\in \text{End}%
(R_{n}^{(s_{n})}),
\end{equation}%
in terms of the form factors $\langle t|O_{n}|t^{\prime }\rangle $. Then
m-point functions can be potentially analyzed by extending and adapting to
the current model the tools that were developed\footnote{%
Note that physical observable as the dynamical structure factors \cite{ARXFBloch36}-\cite{ARXFBalescu75} were accessible by this numerical approach.} in 
\cite{ARXFCM05}-\cite{ARXFCCS07} for the spin 1/2 XXZ quantum chains in the ABA
framework.

We would like to underline that the results presented in this paper provide
a further confirmation of the universality in the characterization of form
factors of local operators of integrable quantum models when described by
our approach in SOV framework. In particular, the comparison of the current
results with those obtained previously by the same approach in \cite{ARXFGMN12-SG}, for the cyclic representations of the trigonometric 6-vertex
Yang-Baxter algebra, in \cite{ARXFN12-0}, for the spin-1/2 h.w. representations
of the same algebra, and in \cite{ARXFN12-2}, for the spin 1/2
representations of the 6-vertex trigonometric reflection algebra, makes
clear this universal picture. Indeed, a part from model dependent features,
like the nature of the spectrum of the quantum separate
variables and the SOV-reconstruction of local operators, the form factors
admit always the same type of determinant representations whose matrix
elements are \textquotedblleft convolutions\textquotedblright\ over
the spectrum of the separate variables of Baxter
equations solutions plus contributions coming from the local operators. The same type of statements will be
proven for the SOS model with antiperiodic boundary conditions and for the XYZ spin-1/2 quantum chain with periodic boundary conditions. In particular, in \cite{ARXFN12-3} the SOV
characterization of the spectrum and the determinant formulae for the scalar products of separate states will be derived allowing in \cite{ARXF?NT12} to get determinant
formulae for the form factors of local operators. It is worth recalling that
in the literature there exist previous results on matrix elements of
local operators which can be related to separation of variable
methods. Of special interest is the Smirnov's paper \cite{ARXFSm98} for the
quantum integrable Toda chain \cite{ARXFSk1} where the form factors of a
conjectured\footnote{%
Note that in Smirnov's paper is the lack of a direct SOV reconstruction of
local operators which forces the use of a conjecture on the basis
of local operators. A
reconstruction has been later presented in \cite{ARXFOB-04} not in the original
Sklyanin's quantum separate variables but it was defined by a change of variables leading to a new set of quantum separate variables.} basis of local operators are
derived in the Sklyanin's SOV framework by determinant formulae which respect
the universal picture above outlined. 

Finally, let us anticipate that the extension of the results derived in this
paper to the antiperiodic XXZ spin $s$ quantum chains is currently under
analysis. Note that while the SOV characterization of the spectrum
and the scalar product formulae can be re-derived trivially along the
same line described in this paper the main technical point to address
remains the solution of the inverse problem for the local spin generators in
a similar form to that presented in Section \ref{ARXFRecLSO}. Indeed, by using
these type of reconstruction the form factors of the local spin generators
will be derived following mainly the same steps presented in Section \ref{ARXFFF}%
. The relevance of this project is evidenced by remarking that the XXZ spin-$s$ quantum chain is strictly related to the lattice discretization of the
sine-Gordon model which is a fundamental prototype of integrable quantum
field theory and then the solution of the XXZ spin-$s$ quantum chain can
provide a further way to solve the sine-Gordon quantum field theory
according to the microscopic approach to integrable quantum field theories
(IQFTs) described in the introduction of \cite{ARXFGMN12-SG}. Our interest
toward integrable microscopic approach to IQFTs is due to need to define an
exact framework where to use the SOV reconstruction of local fields to
overcome the longstanding problem of their identifications\footnote{%
Let us point out that many results are known which confirm the
characterization of massive IQFTs as superrenormalizable perturbations by
relevant local fields \cite{ARXFZam88}-\cite{ARXFGM96} of conformal field theories 
\cite{ARXFVi70}-\cite{ARXFDFMS97}; see for example \cite{ARXFCM90}-\cite{ARXFJMT03} and the
series of works \cite{ARXFDN05-1}-\cite{ARXFDN08}, where the local field content of
massive theories has been classified by that of the associated ultraviolet
CFTs. However, these interesting results on the global operator structure of
massive IQFTs do not solve completely the problem of the
identification of specific local fields in the S-matrix formulation.} in
the S-matrix formulation\footnote{%
See \cite{ARXFA.Zam77}-\cite{ARXFK80} and \cite{ARXFM92} for a review and references
therein.}.

\bigskip 

\textbf{Acknowledgments}\thinspace\ I gratefully acknowledge the YITP Institute of Stony Brook for the opportunity
to develop my research programs under support of National Science Foundation grants PHY-0969739. I would like to thank N. Grosjean, N. Kitanine, K. K. Kozlowski, J.-M. Maillet, B. M. McCoy and V. Terras for their interest in this paper. I would like also to thank the Theoretical
Physics Group of the Laboratory of Physics at ENS-Lyon and the Mathematical
Physics Group at IMB of the Dijon University for their hospitality (under support ANR-10-BLAN-0120-04-DIADEMS). Finally, I would like to thank H. Frahm for pointing out me some important references on the spectrum of higher spin XXX quantum chains and more general related quantum models.
\bigskip

\begin{small}

\end{small}

\end{document}